\documentclass[3p, twocolumn]{elsarticle}

\usepackage{amsthm} \usepackage{amssymb} \usepackage{amsmath} \usepackage{algorithm} \usepackage{algorithmic} \usepackage{graphicx}

\newtheorem{definition}{Definition} \newtheorem{lemma}{Lemma} \newtheorem{theorem}{Theorem} \newtheorem{example}{Example} \newtheorem{corollary}{Corollary} \newtheorem{remark}{Remark}
 \newtheorem{solution}{Solution}

\begin{document}

\begin{frontmatter}

\title{A Near-Optimal Memoryless Online Algorithm for FIFO Buffering Two Packet Classes}

\author{Fei Li}

\address{Department of Computer Science\\
George Mason University\\
Fairfax, Virginia 22030\\
Email: \textsf{lifei@cs.gmu.edu}}


\begin{abstract}
We consider scheduling packets with values in a capacity-bounded buffer in an online setting. In this model, there is a buffer with limited capacity $B$. At any time, the buffer cannot accommodate more than $B$ packets. Packets arrive over time. Each packet is associated with a non-negative value. Packets leave the buffer only because they are either sent or dropped. Those packets that have left the buffer will not be reconsidered for delivery any more. In each time step, at most one packet in the buffer can be sent. The order in which the packets are sent should comply with the order of their arrival time. The objective is to maximize the total value of the packets sent in an online manner. In this paper, we study a variant of this FIFO buffering model in which a packet's value is either $1$ or $\alpha > 1$. We present a deterministic memoryless $1.304$-competitive algorithm. This algorithm has the same competitive ratio as the one presented in (Lotker and Patt-Shamir. PODC 2002, Computer Networks 2003). However, our algorithm is simpler and does not employ any marking bits. The idea used in our algorithm is novel and different from all previous approaches applied for the general model and its variants. We do not proactively preempt one packet when a new packet arrives. Instead, we may preempt more than one $1$-value packet when the buffer contains sufficiently many $\alpha$-value packets.
\end{abstract}


\end{frontmatter}


\section{Introduction}

We consider online algorithms to schedule packets with values in a capacity-bounded buffer. There is a buffer with a limited size $B \in \mathbb Z^+$. At any time, the buffer can accommodate at most $B$ packets. Packets arrive over time. The buffer is preemptive: Packets already in the buffer are allowed to be dropped at any time before they are delivered. We use $r_p \in \mathbb R^+$ and $v_p \in \mathbb R^+$ to denote the \emph{release time} (\emph{arriving time}) and \emph{value} of a packet $p$ respectively. Packets leave the buffer only because they are either sent or dropped. Those sent and dropped packets will not be reconsidered for delivery any more. Time is discrete. In each time step, at most one packet in the buffer can be sent. The order of the packets being sent should comply with the order in which they are released. The objective is to maximize the total value of the packets sent in an online manner. We call this model a \emph{FIFO buffering model}; this model has attracted a lot of attention in the past ten years and has been studied extensively~\cite{KesselmanLMPSS04, LotkerP03, KesselmanMS05, EnglertW09}. In this paper, we study a variant of this model in which packets have value either $1$ or $\alpha > 1$. This variant is called a \emph{two-valued model} and has been investigated in~\cite{KesselmanM03, LotkerP03, EnglertW09}.

Without knowing the future input, an online algorithm has to make decision over time based on the input information that it has seen so far. If an online algorithm decides which packet to send only based on the contents of its current buffer, and independent of the packets that have already been released and processed, we call it \emph{memoryless}. Consider a maximization problem as an example. A deterministic online algorithm is called \emph{$k$-competitive} if its objective value on \emph{any} instance is at least $1 / k$ times of the objective of an optimal offline algorithm applied on the same instance~\cite{BorodinE98}. The \emph{upper bounds} of competitive ratio are achieved by some known online algorithms. A competitive ratio less than the \emph{lower bound} cannot be reached by any online algorithm~\cite{BorodinE98}. For the two-valued model, the previously best known result is a $1.544$-competitive memoryless algorithm~\cite{KesselmanM03}, a $1.304$-competitive algorithm~\cite{LotkerP03} using marking bits to associate with all pending packets in the buffer~\cite{LotkerP03}, and a non-memoryless optimal $1.282$-competitive algorithm~\cite{EnglertW09}. In this paper, we present a $1.304$-competitive memoryless algorithm. Our algorithm is simpler than the one in~\cite{LotkerP03} and it does not use marking bits.

It is instructive to compare and contrast the algorithm in~\cite{LotkerP03} with ours since both have the same competitive ratio $1.304$. Based on the definition of \emph{memoryless algorithms}~\cite{ChrobakJST07, LiSS07} (the algorithm should make the decisions independent of any packets that it has processed), we know that the marking bits used by~\cite{LotkerP03} reflect the packets that the algorithm has processed and affect the marking and flush procedure for later arriving $\alpha$-value packets. Hence the algorithm in~\cite{LotkerP03} is not memoryless.

In Section~\ref{sec:alg}, we describe a deterministic memoryless online algorithm called ON.  In Section~\ref{sec:analysis}, we give the algorithm's analysis, showing that it is $1.304$-competitive. Related work and conclusion remarks are presented in Section~\ref{sec:related}.


\section{Algorithm}
\label{sec:alg}

Without loss of generality, we assume all packets have distinct release time. Consider $m$ packets released in the same time step $t$. We let these $m$ packets have distinct release time of $t$, $t + \delta$, $t + 2 \delta$, $\ldots$, $t + (m - 1) \delta$ respectively, where $\delta > 0$ and $m \cdot \delta \le 1$, in the order of being released.


\subsection{The idea}

The greedy approach might be the first intuitive method to design competitive online algorithms for the FIFO buffering model. It works as follows. If packets overflow, the minimum-value packet is dropped (with ties broken arbitrarily). In each time step, the earliest released packet in the buffer is sent. The greedy algorithm is asymptotically no better than $2$-competitive for the FIFO buffering model, even for the two-valued variant. Based on the observation from the tight instance for the greedy approach, Kesselman et al. in~\cite{KesselmanMS05} came up with another idea, which is to proactively preempt the $1$-value packets in the buffer released before the $\alpha$-value packets. Consider a $1$-value packet $p$ and an $\alpha$-value packet $q$ with $r_p < r_q$. On one hand, if $q$ but not $p$ is the packet sent by the optimal offline algorithm, we would like to preempt $p$ proactively at $q$'s arrival and expect that $q$ can be sent before packet overflow happens. On the other hand, if both $p$ and $q$ are sent by the optimal offline algorithm, we would like to preempt $p$ only if $p$'s value is bounded by a fraction of $q$'s value. This idea leads to the memoryless algorithm PG, which is $1.732$-competitive for the general case~\cite{EnglertW09} and $1.544$-competitive for the two-valued setting~\cite{KesselmanM03}.

In algorithm PG, a packet $p$'s preemption is due to buffering another packet $q$. A new arrival preempts at most one packet that is released earlier. Motivated by the greedy algorithm and PG, we propose the following strategy:
\begin{solution}
We preempt a set of $1$-value packets due to the existence of a set of $\alpha$-value packets in the buffer to make room for the potential $\alpha$-value packets that are released later.
\end{solution}

Different from PG, we take into account the values of multiple packets to preempt a packet. Based on this idea, a $1$-value packet is preempted only when the buffer has buffered sufficiently many of later-released $\alpha$-value packets. In addition, multiple $1$-value packets may be preempted at the arrival of one $\alpha$-value packet.


\subsection{A memoryless online algorithm for the two-valued model}

We name our algorithm ON. ON represents a family of deterministic memoryless online algorithms parameterized by a real number $\beta > 0$. Denote $Q^\text{ALG}_t$ as the buffer of an algorithm ALG at time $t$. Without confusion, we may omit the subscript $t$ in our notation.

\begin{definition}[Ejectable Packet]
Consider two packets $p$ and $q$ in the buffer with $r_p < r_q$, $v_p = 1$ and $v_q = \alpha$. Such a packet $p$ may prevent us from sending $q$ before a future possible packet overflow, and we call $p$ an \emph{ejectable packet}.
\end{definition}

Algorithm ON is outlined as follows. New packets are admitted in a greedy manner. If the earliest-released packet in the buffer is an $\alpha$-value packet, we simply send this packet, same as the greedy policy. Otherwise, if the earliest-released packet is a 1-value packet, we preempt all the ejectable packets, if the total value of all the ejectable packets is bounded by $1 / \beta$ times of the total value of all the $\alpha$-value packets in ON's buffer. Then the earliest-released packet, which is a $1$-value packet if no preemption happens or an $\alpha$-packet if preemption occurs, is sent. In each time step $t$, ON is described in two stages: \emph{admitting packets} (see Algorithm~\ref{alg:admit}) and \emph{(possibly) preempting $1$-value packets and delivering a packet} (see Algorithm~\ref{alg:preempt}).

\begin{algorithm}
\caption{\textsc{Admitting Packets}}
\begin{algorithmic}[1]
\FOR{each new arriving packet}

\IF{there is a buffer slot available}

\STATE append this packet at the end of the packet queue;

\ELSE

\STATE \textbf{\emph{evict}} the minimum-value packet, with ties broken in favor of the earliest-released packet.

\ENDIF

\ENDFOR
\end{algorithmic}
\label{alg:admit}
\end{algorithm}

\begin{algorithm}
\caption{\textsc{Preempting Packets and Delivering a Packet}}
\begin{algorithmic}[1]
\STATE Let the earliest-released packet in the buffer be $e$.

\IF{$v_e = \alpha$}

\STATE send $e$;

\ELSE

\STATE define $D := \{p \ | \ p \in Q^\text{ON}, \ v_p = 1, \ \exists q \mbox{ with } v_q = \alpha \mbox{ and } r_p < r_q\}$;

\IF{$\sum_{q \in Q^\text{ON}; v_q = \alpha} v_q \ge \beta \sum_{p \in D} v_p$}

\STATE \textbf{\emph{preempt}} all the packets in $D$;

\ENDIF

\STATE send the earliest-released packet in the buffer.

\ENDIF
\end{algorithmic}
\label{alg:preempt}
\end{algorithm}

\begin{example}
Let $B = 3$ and $\beta = \alpha$. We use $(r_p, \ v_p)$ to denote a packet $p$ with release time $r_p$ and value $v_p$. Remember that we use fractional release time to differentiate those packets released at the same time step. An input instance is given as follows.

\begin{eqnarray*}
\mbox{step $1$}: & & (1, \ 1), \ (1.1, \ 1), \ (1.2, \ \alpha)\\
\mbox{step $2$}: & & (2, \ \alpha), \ (2.1, \ \alpha), \ (2.2, \ \alpha), \ (2.3, \ 1)\\
\mbox{step $3$}: & & \mbox{no released packets}\\
\mbox{step $4$}: & & \mbox{no released packets}\\
\mbox{step $5$}: & & (5, \ 1), \ (5.1, \ \alpha), \ (5.2, \ \alpha)
\end{eqnarray*}

The optimal offline policy OPT sends the following packets in order:
\begin{displaymath}
(1.2, \ \alpha), \ (2, \ \alpha), \ (2.1, \ \alpha), \ (2.2, \ \alpha), \ (5, \ 1), \ (5.1, \ \alpha), \ (5.2, \ \alpha).
\end{displaymath}

ON sends packet $(1, \ 1)$ in the first time step without preempting the ejectable packets. ON admits all the $\alpha$-value packets released in step $2$ and sends them in steps $2$, $3$, and $4$. In step $5$, ejectable packet $(5, \ 1)$ is preempted. ON sends packets $(5.1, \ \alpha)$ and $(5.2, \ \alpha)$ consecutively in steps $5$ and $6$. Finally, ON has the following sequence of packets being sent.

\begin{displaymath}
(1, \ 1), \ (2, \ \alpha), \ (2.1, \ \alpha), \ (2.2, \ \alpha), \ (5.1, \ \alpha), \ (5.2, \ \alpha).
\end{displaymath}

For the above instance, OPT and ON gain the total values of $6 \alpha + 1$ and $5 \alpha + 1$, respectively.
\end{example}


\section{Analysis}
\label{sec:analysis}

\begin{theorem}
ON is $\max\{\frac{1 + \beta}{\beta}, \ \frac{\alpha^2 + 2 \alpha \cdot \beta}{\alpha^2 + \alpha \cdot \beta + \beta}\}$-competitive, where $\beta > 0$.
\label{theorem:1}
\end{theorem}

We will employ a charging scheme to prove Theorem~\ref{theorem:1}. Let OPT denote an optimal offline algorithm and $\mathcal O$ denote the set of packets sent by OPT.

\begin{lemma}
Any $\alpha$-value packet that ON sends is an $\mathcal O$-packet.
\label{lemma:0}
\end{lemma}

\begin{proof}
Assume there exists an $\alpha$-value packet $p \notin \mathcal O$ that is sent by ON. Using an exchange argument, we will show that there must exist another optimal offline algorithm that sends $p$.

Consider an algorithm MOPT (Modified OPT) which admits packets $\mathcal O \cup \{p\}$ and sends the earliest-released packet in the buffer in each time step. From step $1$ to step $r_p$, MOPT's buffer content and the packet it sends in each time step are the same as those of OPT. We claim that packet overflow must happen in MOPT's buffer at some time step at/after time $r_p$. Otherwise, MOPT can send all the packets in ${\mathcal O} \cup \{p\}$ successfully and gains more than OPT, which contradicts the fact that OPT is optimal. Assume $t \ge r_p$ is the first time at which packet overflow occurs in MOPT's buffer. Since MOPT only accepts one more packet $p$ in addition to the packets $\mathcal O$, there are ($B + 1$) packets for MOPT to buffer at time $t$. We simply drop one packet $q \neq p$ out of MOPT's buffer at time $t$. Because there is no packet overflow in the time interval $[1, \ t]$ and the number of packets buffered by MOPT (after we drop $q$) is the same as that of OPT's at any time after time $t$, MOPT is capable of sending all the packets ${\mathcal O} \cup \{p\} \setminus \{q\}$ in a FIFO order and gains a total value $\ge \sum_{j \in {\mathcal O}} v_j$.
\end{proof}

The contrapositive of Lemma~\ref{lemma:0} leads to the following corollary.

\begin{corollary}
Any non-$\mathcal O$-packet that ON sends is a $1$-value packet.
\label{coro:1}
\end{corollary}

\begin{remark}
From Algorithm~\ref{alg:preempt}, no $\alpha$-value packets can be preempted. That is, any unsent $\alpha$-value packet must have been only evicted by ON.
\label{remark:1}
\end{remark}

\begin{remark}
Consider a time $t$ in which an $\alpha$-value packet is evicted by ON. This must indicate that the current ON's buffer is full of $B$ packets with value $\alpha$. From Algorithm~\ref{alg:preempt}, these $B$ packets with value $\alpha$ will be sent by ON in steps $t, \ t + 1, \ \ldots, \ t + B - 1$.
\label{remark:2}
\end{remark}

\begin{remark}
From Algorithm~\ref{alg:preempt}, if ON preempts some $1$-value packets in a step $t$, ON will send all the preempting $\alpha$-packets in the following time steps. These preempting packets have a total value of at least $\beta$ times of those preempted $1$-value packets.
\label{remark:3}
\end{remark}

In the following, we introduce our charging scheme. Because it is difficult to compare ON with OPT directly, we compare ON with a relaxed algorithm called ROPT. We will show that ROPT gains the same total value as OPT does. In our charging scheme, we will charge values to ROPT and ON. Algorithm ROPT's operations at a time step $t$ is outlined in Algorithm~\ref{alg:ropt}.

\begin{algorithm}
\caption{\textsc{Relaxed OPT} ($p, \ t$)}
\begin{algorithmic}[1]
\STATE Accept each $\mathcal O$-packet arriving at step $t$.

\COMMENT{In Lemma~\ref{lemma:1}, we prove that all the $\mathcal O$-packets can be admitted by ROPT without encountering overflow.}

\COMMENT{Let $p$ be the packet that ON sends in $t$. If ON sends nothing in $t$, $p$ is defined as a non-$\mathcal O$ \emph{null packet} with value $0$.}

\IF{$p \in \mathcal O$ and $p$ is in ROPT's buffer}

\STATE send $p$;

\ELSE

\STATE send the earliest-released packet in the buffer, if any.

\ENDIF
\end{algorithmic}
\label{alg:ropt}
\end{algorithm}

\begin{lemma}
All $\mathcal O$-packets are accepted by ROPT.
\label{lemma:1}
\end{lemma}

\begin{proof}
To prove Lemma~\ref{lemma:1}, we only need to show that at any time ROPT's buffer contains no more pending packets than OPT's buffer, and thus packet overflow does not happen to ROPT when admitting $\mathcal O$-packets. We apply the induction method. Initially, ROPT's and OPT's buffers are empty. Assume at time $t$, the number of packets in ROPT's buffer is no more than the number of packets in OPT's buffer. In step $t$, ROPT sends one packet, if any, and OPT sends one packet, if any. Then after packet delivery, ROPT's buffer still contains no more packets than OPT's buffer.
\end{proof}

\begin{corollary}
ROPT sends all the O-packets.
\label{coro:1+}
\end{corollary}

Given Lemma~\ref{lemma:1} and the fact that OPT successfully sends all the $\mathcal O$-packets, Corollary~\ref{coro:1+} easily holds. Lemma~\ref{lemma:1} and Corollary~\ref{coro:1+} guarantee that
\begin{remark}
In our charging scheme design, we only need to compare ON to ROPT instead of to OPT.
\end{remark}

We describe an important observation of ROPT in Remark~\ref{remark:0}.

\begin{remark}
For any $\mathcal O$-packet $p$ that is sent by ON in step $t$, ROPT either has sent $p$ before $t$ or sends $p$ in the same step $t$.
\label{remark:0}
\end{remark}

\begin{definition}[Chain of Steps]
For a chain consisting of $k$ time steps $$c_1 \rightarrow c_2 \rightarrow \cdots \rightarrow c_k,$$where $c_1 < c_2 < \cdots < c_k$, ON sends a non-$\mathcal O$-packet in step $c_1$ and for all $i = 1, \ \ldots, \ k - 1$, the packet that ROPT sends in step $c_i$ is the packet that ON sends in step $c_{i + 1}$. Note that these time steps do not need to be successive. Chains do not share time steps.
\end{definition}

\begin{lemma}
At any time, for any $\mathcal O$-packet in ON's buffer but not in ROPT's buffer, there is a unique corresponding chain of steps.
\label{lemma:1+}
\end{lemma}

\begin{proof}
In Algorithm~\ref{alg:chain}, we introduce how to build up a chain of steps for each $\mathcal O$-packet $p$ in ON's buffer but not in ROPT's buffer at time $t$. This construction directly proves Lemma~\ref{lemma:1+}. We use $time(p)$ to denote the time step in which ROPT sends a packet $p$.

\begin{algorithm}
\caption{\textsc{Construction of a Chain of Steps ($t$)}}
\begin{algorithmic}[1]
\STATE From Remark~\ref{remark:0}, there exists a unique previous time step $time(p) < t$ in which ROPT sends $p$ and ON sends another packet $q \neq p$.

\IF{$q \notin \mathcal O$}

\STATE create a chain of steps consisting of only one time step $time(p)$.

\ELSE[that is, $q \in \mathcal O$]

\STATE construct a chain of steps $time(q) \rightarrow time(p)$;

\COMMENT{From Remark~\ref{remark:0}, ROPT must send $q$ in a unique time step $time(q) < time(p)$.}

\WHILE{the packet $q'$ that is sent by ON in $time(q)$ is an $\mathcal O$-packet}

\STATE expand the chain by inserting $time(q')$ to the front of the current chain;

\STATE $q$ is replace by $q'$ (for ease of notation of \textbf{while} loop);

\ENDWHILE

\STATE expand the chain by inserting $time(q')$ to the front of the current chain and this chain is completed.

\COMMENT{We have found a non-$\mathcal O$-packet sent by ON and thus the chain is completed, as the head of the chain has to be a non-$\mathcal O$-packet.}

\ENDIF
\end{algorithmic}
\label{alg:chain}
\end{algorithm}
\end{proof}

The basic idea is to start building the chain from the end to the head, in the reverse order of time. The end step is a step that an $\mathcal O$-packet is in ON's buffer but not in ROPT's buffer. Starting from this step, we search backwards in time to look for the step in which ROPT sends this $\mathcal O$-packet. From Remark~\ref{remark:0}, we know that such a step must proceed the end step. Then we look at the packet sent by ON in this step, if it is a non-$\mathcal O$-packet, we stop constructing the chain because we have found the head step of the chain. If it is an $\mathcal O$-packet, we expand the chain and continue to search backwards until we find a non-$\mathcal O$-packet sent by ON.

A characteristics of a chain is that in each time step in the chain except for the first step, ON sends an $\mathcal O$-packet.

\begin{figure*}[htp!]
\includegraphics[width=.74\textwidth]{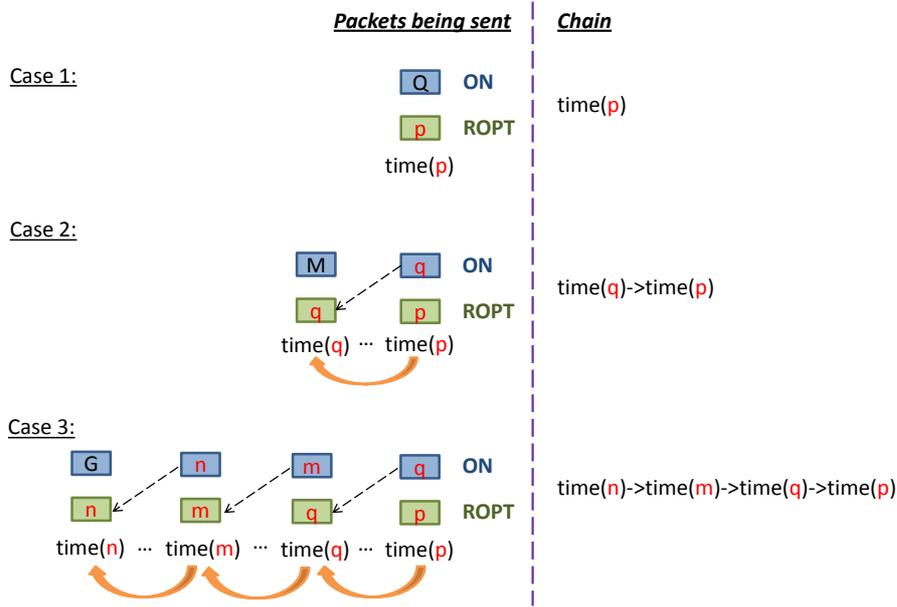}
\caption{The construction of chains of steps}
\label{fig:1}
\end{figure*}

Figure~\ref{fig:1} demonstrates the construction of the chain of steps. Three cases are given. The packets selected by ON and ROPT to send in each step are shown. Capital letter packet is a non-$\mathcal O$-packet, and small letter packet is an $\mathcal O$-packet. The corresponding chains are plotted on the right.

In the following, we introduce our charging scheme for both ON and ROPT. The charging scheme will use the procedure of constructing chains. For those $\mathcal O$-packets that are preempted or evicted by ON, we construct the chains at the time steps when they are preempted or evicted. Then, we charge these $\mathcal O$-packets' values to ROPT either at the first time steps of the chains or in the time steps after they are preempted or evicted. Details of the charging scheme are described as follows.

\begin{definition}
\textbf{Open/closed chain.} Given a time $t \ge c_k$, consider a chain of $k$ time steps $c_1 \rightarrow c_2 \rightarrow \cdots \rightarrow c_k$. We call this chain \emph{closed} if we have charged the value of a packet in ROPT's current buffer to ROPT in step $c_1$. (Note that in step $c_1$, ON sends a non-$\mathcal O$-packet; from Corollary~\ref{coro:1}, this packet is a $1$-value packet.) Otherwise, we say that this chain is \emph{open}.
\end{definition}


\subsubsection*{Case $1$.}

For each packet sent by ON, we charge ON the value of this packet in the time step that it is being sent.


\subsubsection*{Case $2$.}

For those $\mathcal O$-packets that are sent by both ROPT and ON, we charge their values to ROPT in the time steps that ON sends them.


\subsubsection*{Case $3$.}

For the $\mathcal O$-packets that are not sent by ON, they are either evicted or preempted.


\paragraph{Assume the unsent packet is an $\alpha$-value packet}

Those $\alpha$-value $\mathcal O$-packets that are not delivered by ON can only be evicted (from Remark~\ref{remark:1}).

Let $p$ denote an evicted $\alpha$-value $\mathcal O$-packet and $d_p$ denote the time in which ON evicts $p$. From Remark~\ref{remark:2}, ON sends at least $B$ packets with value $\alpha$ after time $d_p$. We define $l_p$ as the first time step after $d_p$ such that in step $l_p$ ON sends either a $1$-value packet or nothing. We charge the value $v_p = \alpha$ to ROPT in the time interval $[d_p, \ l_p - 1]$.


\paragraph{Assume the unsent packet is a $1$-value packet}

Those $1$-value $\mathcal O$-packets that are not delivered by ON can be evicted or preempted by ON.

\begin{enumerate}
\item Assume $p$ is a preempted $1$-value $\mathcal O$-packet and $d_p$ is the time in which ON preempts $p$.

\begin{enumerate}
\item Assume there is an open chain associated with a preempting packet at time $d_p$.

Let $q$ be the earliest-released $\alpha$-value preempting packet in ON's buffer whose corresponding chain is open. From Algorithm~\ref{alg:chain}, $q$ is an $\mathcal O$-packet that has been sent by ROPT in a previous time step $< d_p$. In the first step of this chain, say $t'$, ON sends a $1$-value non-$\mathcal O$-packet (see Corollary~\ref{coro:1}).

We charge the value $v_p = 1$ to ROPT in step $t'$ and we close this chain.

\item Assume there is no open chain associated with any preempting packet at time $d_p$.

Let $h \ge 1$ be the number of preempting packets in ON's buffer at time $d_p$. From Remark~\ref{remark:3}, ON sends these preempting $\alpha$-packets consecutively from $d_p$ to $d_p + h - 1$.

We charge the value $v_p = 1$ to ROPT in each of the time step in the interval $[d_p, \ d_p + h - 1]$.
\end{enumerate}

\item Assume $p$ is an evicted $1$-value $\mathcal O$-packet and $d_p$ is the time in which ON evicts $p$.

\begin{enumerate}
\item Assume $p$ is an $\mathcal O$-packet that has been sent by ROPT by time step $d_p$.

From Lemma~\ref{lemma:1+}, $p$ corresponds to a chain of steps such that $p$ is the packet sent by ROPT in the last time step of the chain and in the first time step of the chain, say, $t'$, ON sends a $1$-value non-$\mathcal O$-packet (see Corollary~\ref{coro:1}). Because $v_p = 1$, $p$ is not a preempting packet and no preempted $1$-value packet has been charged in the first step $t'$ of the chain for ROPT.

We charge the value $v_p = 1$ to ROPT in step $t'$ and we close this chain.

\item Assume $p$ is rejected by ON at its arrival. ROPT accepts $p$ and will send $p$ in a later time step $\ge d_p$.

From Algorithm~\ref{alg:admit}, if this case happens, it must be true that ON's buffer is full of $\alpha$-value packets at $p$'s arrival.

We first claim that all the packets in ON's buffer at time $d_p$ are $\mathcal O$-packets. Because otherwise, ROPT can use an $\alpha$-value non-$\mathcal O$-packet which is only in ON's buffer to replace $p$ ($v_p = 1$) in ROPT's buffer to gain more value. We then claim that there must exist at least one open chain at time $d_p$ since otherwise, each closed chain corresponds to one packet in ROPT's buffer which forbids ROPT to accept $p$.

Let $q$ be the earliest-released $\alpha$-value $\mathcal O$-packet in ON's buffer whose corresponding chain is open. In the first step of this chain, say $t'$, ON sends a $1$-value non-$\mathcal O$-packet (see Corollary~\ref{coro:1}).

We charge the value $v_p = 1$ to ROPT in step $t'$ and we close this chain.
\end{enumerate}
\end{enumerate}

\begin{remark}
For any evicted $1$-value $\mathcal O$-packet, its value is charged to ROPT in a time step $t'$ in which ON sends a $1$-value non-$\mathcal O$-packet and $t'$ is the first time step of a closed chain of steps.
\label{remark:4}
\end{remark}

\begin{remark}
For each preempted $1$-value $\mathcal O$-packet, if its value is charged to ROPT in a time step $t'$ in which ON sends a $1$-value non-$\mathcal O$-packet and $t'$ is the first time step of a closed chain of steps, the gain ratio in this time step $t'$ is bounded by $1$.
\label{remark:4+}
\end{remark}

Remark~\ref{remark:4} and Remark~\ref{remark:4+} indicate that in the time steps that we charge $1$-value evicted/preempted packets to ROPT, the gain ratio is bounded by $1$. In the time step when ON sends an $\mathcal O$-packet, the value of the $\mathcal O$-packet is charged to ROPT in the same time step and the gain ratio is $1$. Thus, in order to prove Theorem~\ref{theorem:1}, we only need to analyze the gain ratio for the evicted $\alpha$-values $\mathcal O$-packets and the preempted $1$-value $\mathcal O$-packets. (Recall that Remark~\ref{remark:1} shows that no $\alpha$-value packet is evicted by ON.)

\begin{remark}
Each evicted (respectively, preempted) $\mathcal O$-packet $p$ is associated with a time interval $[d_p, \ l_p - 1]$ (respectively $[d_p, \ d_p + k - 1]$). In the time steps falling in these intervals, ON sends $\alpha$-value packets only.
\label{remark:5}
\end{remark}

To avoid double-charging the $\mathcal O$-packets unsent by ON, we have the following results.

\begin{lemma}
Consider an interval in which ON sends preempting $\alpha$-value packets. If there are preempted packets that are charged to ROPT in this interval, then no evicted $\alpha$-value packets are charged to ROPT in this interval.
\label{lemma:1++}
\end{lemma}

\begin{proof}
Note that if a preempted $1$-value $\mathcal O$-packet $p$ is charged to ROPT in this interval, then at time $d_p$ by when ON preempts $p$, there are no open chains. From time $d_p$ to the time when ON sends all the preempting packets, no new chains are generated and no open chains exist. Also, for each closed chain, if its last packet is in ON's buffer, the first time step of this chain corresponds to a unique packet in ROPT's buffer. Hence, no $\alpha$-value packet will be evicted by ON during this interval.
\end{proof}

\begin{corollary}
Consider an interval in which ON sends $\alpha$-value packets. If there are evicted $\alpha$-value packets charged to ROPT in this interval, then no preempted $1$-value packets are charged to ROPT in this interval.
\label{coro:2}
\end{corollary}

Given Lemma~\ref{lemma:1++} and Corollary~\ref{coro:2}, to prove Theorem~\ref{theorem:1}, we will show that
\begin{enumerate}
\item in each interval ON sends preempting $\alpha$-value packets, the total value of the preempted $1$-value packets assigned to ROPT is bounded by $\frac{1}{\beta}$ times of the value gained by ON; and

\item in each interval ON sends $\alpha$-value packets, if $x$ evicted $\alpha$-value $\mathcal O$-packets are charged to ROPT in this interval, then that there are $x$ open chains corresponding to these $\alpha$-value packets. Also, there are $x$ time steps in which ON sends non-$\mathcal O$-packets and no values are charged to ROPT in those time steps. Also, $x$ is always bounded by $\frac{\beta}{\alpha + \beta}$.
\end{enumerate}

In the following, we consider the gain ratio for an interval with evicted $\alpha$-value $\mathcal O$-packets.

\begin{lemma}
At any time, the number of $\mathcal O$-packets which are in ON's buffer but have been sent by ROPT is no more than $\frac{B \cdot \beta}{\alpha + \beta}$.
\label{lemma:2}
\end{lemma}

\begin{proof}
In ON's buffer, the cumulative number of $\alpha$-value packets that have been sent by ROPT is increased by $1$ only when ON sends a $1$-value $\mathcal O$-packet. That means no preemption happens in that time step; otherwise, that $1$-value $\mathcal O$-packet will be preempted and an $\alpha$-value preempting packet will be sent.  For each time step that ON sends a $1$-value non-$\mathcal O$-packet, we have the following inequality (let $x$ be the cumulative number of evicted $\alpha$-value $\mathcal O$-packets): $\alpha < (B - 1) \beta$, $2 \alpha < (B - 2) \beta$, $\cdots$, $x \cdot \alpha < (B - x) \beta$. From $x \cdot \alpha < (B - x) \beta$, we have $x < \frac{B \cdot \beta}{\alpha + \beta}$.
\end{proof}

The above inequality limits the number $x$ of cumulative evicted $\alpha$-value $\mathcal O$-packets that we charge to ROPT in the time interval that ON sends at least $B$ packets with value $\alpha$. From Lemma~\ref{lemma:2}, $x < \frac{B \cdot \beta}{\alpha + \beta}$. For each evicted $\alpha$-value packet, it corresponds to a time step in which ON sends a $1$-value non-$\mathcal O$-packet and in that time step, we do not charge ROPT any value. Thus, in those $x$ time steps and this time interval, ROPT gains a total value of $(x + B') \alpha$ and ON gains a total value of $x + B' \cdot \alpha$, where $B' \ge B$. Then we have the ratio of gains bounded by (note $B' \ge B$ and $x < \frac{B \cdot \beta}{\alpha + \beta}$)
\begin{eqnarray*}
\frac{(x + B') \alpha}{x + B' \cdot \alpha} & \le & \frac{(x + B) \alpha}{x + B \cdot \alpha}\\
& \le & \frac{\left(\frac{B \cdot \beta}{\alpha + \beta} + B\right) \alpha}{\frac{B \cdot \beta}{\alpha + \beta} + B \cdot \alpha}\\
& = & \frac{\left(\frac{\beta}{\alpha + \beta} + 1\right) \alpha}{\frac{\beta}{\alpha + \beta} + \alpha}\\
& = & \frac{\alpha^2 + 2 \alpha \cdot \beta}{\alpha^2 + \alpha \cdot \beta + \beta}.
\end{eqnarray*}

For the interval in which ROPT is charged with preempted $1$-value packets, we know that there are no evicted $\alpha$-value packets are charged to ROPT in this interval (see Lemma~\ref{lemma:1++}). Thus, the total value of the preempted $1$-value $\mathcal O$-packets is bounded by $\frac{1}{\beta}$ times of the total value of the $\alpha$-value preempting packets. Note that these preempting $\alpha$-value packets may be $\mathcal O$-packets and we charge their values to ROPT in these time steps, thus, the gain ratio is bounded by $\frac{1 + \beta}{\beta}$.

We want to minimize the gain ratio $\rho$ for all the time intervals, where $\rho = \max\{\frac{1 + \beta}{\beta}, \ \frac{\alpha^2 + 2 \alpha \cdot \beta}{\alpha^2 + \alpha \cdot \beta + \beta}\}$. In order to get $\rho = \frac{1 + \beta}{\beta}$ for any $\alpha$, we have that for any $\alpha$, $\frac{1 + \beta}{\beta} \ge \frac{\alpha^2 + 2 \alpha \cdot \beta}{\alpha^2 + \alpha \cdot \beta + \beta}$. This requires
\begin{equation}
\alpha^2 - \beta (\beta - 1) \alpha + \beta^2 + \beta > 0.
\label{equ:final}
\end{equation}

To satisfy the inequality in Equation~\ref{equ:final} for any $\alpha$, we need to guarantee $\beta^2 (\beta - 1)^2 - 4 (\beta^2 + \beta) = \beta\left(\beta^3 - 2 \beta^2 - 3 \beta - 4\right) < 0$. By solving $\beta^3 - 2 \beta^2 - 3 \beta - 4 < 0$, we have $\beta \le 3.284$. Hence, we get the gain ratio $\rho$ minimized at $1.304$ when $\beta = 3.284$.

\begin{corollary}
ON is $1.304$-competitive when $\beta = 3.284$.
\end{corollary}


\section{Related Work and Open Problems}
\label{sec:related}

Mansour et al.~\cite{MansourPL04} initiated the study of competitive online algorithms for the FIFO buffering model. They designed a simple greedy deterministic algorithm with a tight competitive ratio $2$~\cite{KesselmanLMPSS04}. The first algorithm with a competitive ratio strictly less than $2$ was presented by Kesselman et al.~\cite{KesselmanMS05}. Englert and Westermann~\cite{EnglertW09} showed that PG is $1.732$-competitive but no better than $1.707$-competitive. The lower bound of competitive ratio for deterministic algorithms is $1.409$~\cite{KesselmanMS05}. For the two-valued variant in which packets have value either $1$ or $\alpha > 1$, Kesselman and Mansour~\cite{KesselmanM03} proposed a $1.544$-competitive memoryless algorithm. Englert and Westermann~\cite{EnglertW09} presented an optimal $1.282$-competitive algorithm which meets the lower bound~\cite{KesselmanLMPSS04}. However, this algorithm~\cite{EnglertW09} is not memoryless.

In this paper, we present a $1.304$-competitive memoryless algorithm for the two-valued variant. In~\cite{LotkerP03}, an algorithm using marking bits achieves the same competitive ratio $1.304$. The algorithm that we present in this paper is simpler and it is not using marking bits. All previous work proactively preempt packets. On the contrary, our algorithm drops packets in a `lazy' manner. For this variant, closing or shrinking the gaps of $[1.282, \ 1.304]$ for memoryless algorithms remains an open problem. We hope that the idea presented in this paper will motivate an improved algorithm for the general FIFO model with arbitrary values.



\bibliographystyle{elsarticle-num}
\bibliography{complete}


\end{document}